\documentclass[a4paper,twocolumn,aps,prd,superscriptaddress,nofootinbib]{revtex4-2}

\usepackage{amsmath,amssymb,amsfonts,mathrsfs,amsthm}

\usepackage[colorlinks,citecolor=blue,linkcolor=blue,urlcolor=blue, breaklinks=true]{hyperref}

\usepackage{mathtools}

\usepackage{enumitem} 

\usepackage{bbold}
\usepackage{bbm} 
\usepackage{physics} 

\usepackage[english]{babel} 

\usepackage{tikz-cd} 

\usepackage{xcolor}
\usepackage[normalem]{ulem}

\newtheorem{lemma}{Lemma}

\newcommand{\passage}[1]{\bigskip
                         {\bf #1.}
                         \nopagebreak}

\begin{document}

\title{Fact-nets: towards a mathematical framework for relational quantum mechanics}

\author{Pierre Martin-Dussaud}
\email{martindussaud@gmail.com}
\affiliation{Institute for Gravitation and the Cosmos, The Pennsylvania State University, University Park, Pennsylvania 16802, USA}
\affiliation{Basic Research Community for Physics e.V., Mariannenstra\ss e 89, Leipzig, Germany}
\author{Titouan Carette}
\email{titouan.carette@universite-paris-saclay.fr}
\affiliation{Basic Research Community for Physics e.V., Mariannenstra\ss e 89, Leipzig, Germany}
\affiliation{Université Paris-Saclay, Inria, CNRS, LMF, 91190, Gif-sur-Yvette, France}
\author{Jan Głowacki}
\email{glowacki@cft.edu.pl}
\affiliation{Basic Research Community for Physics e.V., Mariannenstra\ss e 89, Leipzig, Germany}
\affiliation{Center for Theoretical Physics of the Polish Academy of Sciences, Al. Lotników 32/46 02-668 Warsaw, Poland}
\author{Vaclav Zatloukal}
\email{zatlovac@gmail.com}
\affiliation{Faculty of Nuclear Sciences and Physical Engineering, Czech Technical University in Prague, \\
B\v{r}ehov\'{a} 7, 115 19 Praha 1, Czech Republic}
\author{Federico Zalamea}
\affiliation{Basic Research Community for Physics e.V., Mariannenstra\ss e 89, Leipzig, Germany}
\affiliation{EFREI Research Lab, EFREI Paris}

\date{ \small\today}

\begin{abstract}
The relational interpretation of quantum mechanics (RQM) has received a growing interest since its first formulation in 1996. Usually presented as an interpretational layer over the usual quantum mechanics formalism, it appears as a philosophical perspective without proper mathematical counterparts. This state of affairs has direct consequences on the scientific debate on RQM which still suffers from misunderstandings and imprecise statements. In an attempt to clarify those debates, the present paper proposes a radical reformulation of the mathematical framework of quantum mechanics which is relational from the start: fact-nets. The core idea is that all statements about the world, facts, are binary entities involving two systems that can be symmetrically thought of as observed and observer. We initiate a study of the fact-nets formalism and outline how it can shed new relational light on some familiar quantum features.
\end{abstract}

\maketitle

\section{Introduction}

The standard formulation of quantum mechanics relies heavily on the framework of Hilbert spaces. It is an effective piece of mathematics that encodes the physical notion of a pure state and provides ground for the superposition principle.

Historically, the introduction of the notion of a `state' understood as a wave-function came as a relief to the unease of matrix mechanics. Schrödinger recovered the results of the latter by solving a differential equation, rather than using linear algebra. Besides, the idea of a wave seemed to provide a useful image to understand quantum phenomena. It continued de Broglie's idea of waves of matter \cite{schrodinger1926c}.

However, this intuition breaks down whenever multiple particles are considered. Indeed, the wave-function is not a function over space anymore, but instead a function over the configuration space. This was already pointed out sadly by Schrödinger at the Solvay conference in 1927 (\cite{bacciagaluppi2006} p. 447). 

Another trouble is coming from the so-called measurement problem, for instance in the context of the Wigner's friend thought experiment. It can be stated as the incompatibility between the statements: 
\begin{enumerate}
    \item All physical systems are quantum.
    \item The wave-function is observer-independent.
    \item The evolution of an isolated system is unitary.
    \item A measurement is a collapse of the wave-function.
\end{enumerate}
The Copenhagen interpretation rejects the first, relational quantum mechanics (RQM) the second, collapse models the third and many-worlds the fourth statement.

Here we stick to a relational interpretation, which implies that the wave-function is only a useful bookkeeping device relative to an observer, which can be any other system. This stance is surprising for anyone who was taught quantum mechanics \textit{à la} Schrödinger, like in von Neumann's book where it is said
\begin{quote}
    \textit{``it is evident that everything which can be said about the state of a system must be derived from its wave function"} \cite{vonneumann2018}.
\end{quote}
It is said implicitly that the wave function is an absolute physical quantity, i.e. pertaining only to a system in itself, without any reference to an observer.
In this perspective, the wave-function is the new entity that gathers the properties of a system and the challenge consists in solving the Schrödinger equation, which describes the probabilistic evolution of these properties. 

Relational quantum mechanics points out the important role of the observer. The line of separation between the system and the observer is usually called the Heisenberg cut. Relational quantum mechanics insists on the idea that this cut can be located between any two systems. In other words, any system can be taken as a reference, and thus promoted to an observer. This departs from the usual use of the word `observer' which is often loaded with the idea that the system should be macroscopic. 

This stance shifts the focus from the system itself to the relation between an observer and the system. In the measurement problem, Wigner and his friend, after the friend has measured the system, do not attribute the same state to the system. So the state is not a property of the system itself, but a description of a relation between the system and another system called the observer. The relative state encodes information obtained through past interactions between the system and the observer. The interaction drives a `quantum event' or a `fact' that persists in the form of a correlation between the system and the observer. 

Since its initial formulation in 1996 \cite{rovelli1996b}, the relational interpretation has evolved. It is not yet a fixed dogma, but still a living flow of ideas stemming from the same initial corpus, and trying to merge into a consistent common vision of the world. Many questions certainly remain to be answered \cite{laudisa2017,pienaar2021,brukner2021,dibiagio2021a,stacey2021a,adlam2022}. The objections which have been formulated against it mainly arise from the unfortunate mixture between standard quantum mechanical conceptions and relational statements. This confusion can be partly explained by the use of the same mathematical formalism, so that the relational interpretation distinguishes itself from others only on the level of the interpretation of the formalism.

We believe that some confusion would be waived if the distinction from other interpretations was brought to the level of the formalism, so that it would become more than an interpretation, but indeed a theory. A formalism, like a grammar, imposes some constraints upon what is even expressible. So far, the relational interpretation has to deal with a standard format of quantum mechanics, which is not well-adapted. The challenge would be to reformulate quantum mechanics with a better-suited language, where the relational aspect is already implemented in the basic grammar. In this article, we propose a route towards such a formalization. Similar efforts have been recently pushed in this direction \cite{yang2018}, but our approach is essentially independent. Here, we start by building upon the idea, present in the most recent formulation of the relational interpretation, that the primordial entities are the relative facts \cite{dibiagio2021, dibiagio2021a}. 

This article organizes as a pedestrian investigation that slowly builds a consistent framework to incorporate step by step what we know about quantum mechanics. It subdivides into the following sections:
\begin{itemize}
    \item [\ref{sec:facts}.] \textbf{Facts}: we explain what a relative fact is, and we define \textit{fact-nets}.
    \item [\ref{sec:amplitude}.]  \textbf{Amplitudes}: we define the \textit{amplitude}, investigate its properties, its physical meaning and show how it can be used to compute conditional probabilities.
    \item [\ref{sec:composite_systems}.] \textbf{Composite systems}: we show how fact-nets can account for composite systems.
    \item [\ref{sec:measurement}.] \textbf{Measurement}: we show how measurements transform fact-nets.
    \item [\ref{sec:recovering}.] \textbf{Recovering the Hilbert space formalism}: we show how the standard formalism of quantum mechanics can be recovered.
    \item [\ref{sec:representing}.] \textbf{Self-space of a system}: we show an alternative representation of systems in fact-nets in terms of Hilbert spaces.
    \item [\ref{sec:amplitude_maps}.] \textbf{Amplitude maps}: we show how the amplitude can be captured as a map between relative Hilbert spaces, and fact-nets as diagrams.
    \item [\ref{sec:QRF}.] \textbf{Quantum reference frames}: we show how changes of quantum reference frames can be described within a triangle fact-net.
\end{itemize}

\section{Facts}
\label{sec:facts}

The assumed ontology of relational quantum mechanics was not etched in stone in the first paper \cite{rovelli1996b}. In its early stages, the key concept seemed to be the information exchanged between systems. In its most recent formulation, the ontology is based on \textit{facts} or \textit{quantum events}, both expressions being used interchangeably. We stick to the conception of facts as presented in \cite{dibiagio2021a}:
\begin{quote}
    \textit{[Facts] happen in interactions between any two systems and can be described as the actualization of the value of a variable of one system relative to the other.}
\end{quote}
The term `quantum event' should maybe be preferred because it is discrete, ephemeral, sparse, while `fact' has a connotation of something more absolute or enduring in time. But the term `quantum event' is also loaded with a space-time connotation, which is not really wanted. So we stick to the name of fact, which is already more spread in the literature, presents the morphological advantage of being short and sounds like the brief detonation it must be. 

A fact persists in the form of a correlation between systems. One of the crucial claims at the basis of relational quantum mechanics is that measurements are not ontologically different from interactions between two systems: a measurement of system $A$ by system $B$ is nothing more than an interaction between systems $A$ and $B$ viewed from the perspective of $B$. Therefore, from the point of view of an observer, i.e. any of the two reference systems implied in the interaction, the interaction is a measurement and a fact is an outcome of the measurement. Beware that the terms of observers and measurements are often loaded with additional assumptions, like macroscopicity. We do not make such an assumption here. In our language, the same fact can be regarded symmetrically as the measurement of one system by another or vice-versa.

The facts are sometimes completed with an epithet as \textit{relative facts} to underline that facts are relative to systems. In \cite{dibiagio2021}, it is said that facts are \textit{relative to one} system. However, since facts are also \textit{about one} system, it actually requires \textit{two} systems to have a fact. So facts must be labelled by the two systems which interact, and they can be described alternatively as a fact \textit{about} one system \textit{relative to} the other, or vice-versa.

The intuitive picture that emerges is that of a network of systems related by facts. Thus, we will assume that the basic structure of any physical situation can be represented by a graph, with systems as nodes and sets of possible facts as links in-between. For instance, consider
\begin{equation}\label{first}
   \begin{tikzcd}
   & D  \arrow[d, "x", no head, shift left] \arrow[d,  no head, shift right] \arrow[d,  no head]  & \\
A \arrow[r, "y", no head] \arrow[r, no head, shift right]  
& B \arrow[r, "z",  no head]
\arrow[r, no head, shift right]
& E \\
& C \arrow[u, "w", no head] \arrow[ru, "v"', no head] \arrow[ru, no head, shift right] &
\end{tikzcd} 
\end{equation}
There are five systems A, B, C, D, E. Between two systems, a set of lines represent a set of possible facts. The variables $v,w,x,y,z$ label the possible facts between two systems.

It is important to notice the move from more standard descriptions of physics. The usual grammar (some might say ontology) of physics has systems (humans, bananas) that carry properties (tall, yellow). Here, systems do not have properties. Instead, there are facts between systems. It is only in special cases that these facts can be understood as properties of a system. Besides, no foundational hierarchy is set between `systems' and `relations': they co-appear in fact-nets.

This presentation is contrasting with what is often taken as an axiom of quantum mechanics: `the state of a system is described by a vector in a Hilbert space'. The latter formulation pertains to the usual grammar of physics, where the state gathers the properties of a system. This is hiding the role of the observer, which is only implicit. Bringing the observer back on stage suggests our change of grammar.

\passage{Definition 1: fact-nets}

The core of our formulation of relational quantum mechanics are fact-nets, a synthetic way to present the multiple potential results of interactions between systems.

\medskip

\textit{A \textbf{fact-net} $F$ is a multi-graph\footnote{A multi-graph is a graph where we allow more than one edge between two vertices.} whose vertices are called \textbf{systems} and edges \textbf{facts}. The \textbf{fact-set} of edges between two systems $A$ and $B$ is denoted $F_{AB}$. The set of all edges incident to one vertex $A$ is denoted $F_A$. We say that a fact $f$ \textbf{involves} a system $A$ when $f \in F_A$, so that $F_A$ is the set of all facts involving $A$. We say that the system $A$ \textbf{is related to} the system $B$ if the fact-set $F_{AB}$ is non-empty. We call a system $A$ \textbf{finite} if $F_A$ is such. The set of all systems related to $A$ is denoted $N_A$ and called the \textbf{neighbourhood} or \textbf{environment} of $A$.}

\medskip

We see that, in a fact-net, the role of the systems is minimal: they are only here to organize facts. Then, any property assigned to a system can only come from the facts involving it. A fact always involves two systems, which can be interpreted as an observer and an observed or vice-versa. This is the implementation of the foundational principle of relational quantum mechanics that nothing should be said about a system without an explicit reference to the observer. The edges are not oriented because we assume the relations to be symmetric in the sense of involving both systems in the same way, akin to the action/reaction principle. Several edges are allowed between two systems to account for multiple possible facts.

One may wonder whether fact-nets themselves are relative to a choice of observer. In \cite{adlam2022}, it is argued that \textit{the existence of a fact} is an absolute fact, i.e. observer-independent, while only \textit{the value of that fact} is relative. This would imply that the fact-net also is absolute. However, we can be agnostic about this to use and develop the fact-net formalism. It is perfectly possible to think that fact-nets themselves only make sense with respect to an observer.

\section{Amplitude}
\label{sec:amplitude}

A reason why Hilbert spaces have become central in standard introductions of QM is that they are the natural mathematical object to formalize the superposition principle. However, the superposition principle brings on board some doubtful preconceptions about what is real. In popular science, a superposed state of positions is often interpreted as the particle being in two places at the same time. Most physicists recognize this statement is too rough. It is more appropriate to say that the state is a superposition of \textit{potentialities}, while the position of the particle is not definite yet. Understood as such, the superposition principle is a statement about correlations between earlier and later measurements, rather than a statement about the actual properties of the system. This view remains closer to the experimental content of the principle.

Let us now show how a superposed state looks in the fact-net formalism.

\passage{Example 1: two directions of spin}

Consider two Stern-Gerlach apparatuses $\mathfrak{S}_z$ and $\mathfrak{S}_x$, respectively oriented along the $z$ and $x$ axes. They both interact with a $1/2$-spin particle $S$. For each direction, $x$ or $z$, there are two possible facts, up or down, denoted $(0,1)$ for $\mathfrak{S}_z$ and $(+,-)$ for $\mathfrak{S}_x$. This is summarized as
\begin{equation}\label{2SG}
    \begin{tikzcd}
\mathfrak{S}_z \arrow[r, "0", no head, shift left] \arrow[r, "1"', no head, shift right] 
& S \arrow[r, "+", no head, shift left]  \arrow[r, "-"', no head, shift right]
& \mathfrak{S}_x \\
\end{tikzcd}
\end{equation}
It is a matter of experiment that the facts $\{0,1\}$ and $\{-,+\}$ are correlated. In the standard formalism of quantum mechanics, these facts would be understood as properties of the particle (spin $z$ and spin $x$) and expressed as states related by
\begin{equation}\label{eq:schroedinger_cat}
    \ket{\pm} = \frac{1}{\sqrt{2}} \left( \ket{0} \pm \ket{1} \right).
\end{equation}

Hence the idea that a cat can be both dead and alive, or that a particle can occupy two positions at the same time. But this is a very loose way of speaking, relying on the metaphysical postulate that systems have properties independently of them being measured. The spin can certainly be observed in the spin $+$ but there is a conceptual leap to say that this is a superposition of $0$ and $1$. What is actually observed are only facts: $+$,$-$, $0$ or $1$ and the conditional probabilities between them. So what matters for the completeness of our physical theory is the ability to compute the conditional probabilities between different facts. This is done with the help of an amplitude $W$, which takes two facts as input and gives back a complex number. In our example we have
\begin{equation}\label{eq:amplitude_2SG}
\begin{split}
    &W_{S}(0,\pm) = \frac{1}{\sqrt{2}} \\
    &W_{S}(1,\pm) = \pm \frac{1}{\sqrt{2}}.
\end{split}
\end{equation}
These equations have the same information content as equation \eqref{eq:schroedinger_cat} but present it in a way more adapted to the grammar of facts: an amplitude evaluates a degree of interdependence between facts involving a common system. This can be translated in terms of conditional probability. For instance
\begin{equation}
    P_{{S}}(0|+) = \frac{|W_{{S}}(0,+)|^2}{|W_{{S}}(0,+)|^2 + |W_{{S}}(1,+)|^2} = \frac{1}{2}.
\end{equation}

\passage{Example 2: position/momentum}

Let's provide another example. Consider a particle $S$, a ruler $X$ and an apparatus $P$ that measures momentum. The picture is the following
\begin{equation}
   \begin{tikzcd}
X \arrow[r, no head, shift right] \arrow[r, no head] \arrow[r, "x", no head, shift left]
& S \arrow[r, no head, shift right] \arrow[r, no head] \arrow[r, "p", no head, shift left]
& P 
\end{tikzcd} 
\end{equation}
All the physics of the situation can be computed from the amplitude
\begin{equation}
    W_S (x,p) = e^{\frac{i}{\hbar} p x}.
\end{equation}
$W$ is a tool that enables us to interpret some facts in relation to others. In the two examples above, the systems in the neighbourhood $N_S$ play the role of measurement apparatuses associated to observables defining sets of possible facts about $S$. In the case of the position/momentum, the two observables have different physical units, so the amplitude contains a conversion factor, $\hbar$. 

\passage{Definition 2: amplitude}

Here we come to generalize the previous examples to general fact-nets. 

\medskip

\textit{A quantum theory over a fact-net consists in describing the pairwise correlations between facts. Formally, for each system $S$, we have a function, called the \textbf{amplitude}, $W_S : F_S \times F_S \to \mathbb{C}$ that takes two facts involving $S$ and gives a complex number. The physical interpretation of the amplitude is given in terms of \textbf{conditional probability} as 
\begin{equation}
\label{eq:conditional_probability}
    P_{{S}} (a|b) =\frac{|W_S (a,b)|^2}{\sum_{a' \in F_{SA}} |W_S (a',b)|^2},
\end{equation}
with $a \in F_{SA}$ and $b \in F_{SB}$.
}

\medskip

This probability is an epistemic quantity relative to the common system $S$ involved by the two facts. The probability $ P_{S} (a|b) $ is interpreted as the conditional probability, from the perspective of $S$, of fact $a$ to be actualized knowing that fact $b$ has already been actualized.

This formula is analogous to Born's rule. It is the bridge from the mathematical formalism to the physical tests of the theory. One can check that the formula produces a real number between $0$ and $1$. Note that when the denominator vanishes, the numerator vanishes too, so that the conditional probability is not well-defined in such a case.

Let's point out that in this framework, the physical probabilities are always \textit{conditional probabilities}, which means that the probability of a fact is always relative to another fact. This point of view is in-line with the general philosophy of RQM: conditional probabilities are the expected consequence of considering only relative properties. 

In the standard formalism, the amplitude can usually be computed from the states as 
\begin{equation}
   W_S(a,b) = \braket{a}{b}.
\end{equation}
But our framework shifts the focus from the states to the amplitude so that we want to define the amplitude without introducing the states. By doing so, we retain all the physically relevant content.

We will now investigate, through examples, what are the generic properties to be expected of $W$.

\passage{Hermiticity}

First, we require $W$ being \emph{hermitian} or \textit{conjugate symmetric}:
\begin{equation}\label{eq:hermitian}
    \forall f,g \in F_S,~ W_S (g,f)=\overline{W_S (f,g)}.
\end{equation}
This property expresses a fundamental symmetry between the facts. However, this symmetry does not get through to the conditional probability because of the denominator in equation \eqref{eq:conditional_probability}, so that in general we have 
\begin{equation}
    P_{{S}} (a|b) \neq P_{{S}} (b|a).
\end{equation}
For instance, consider the following fact-net
\begin{equation}
    \begin{tikzcd}
A \arrow[r, "a_0", no head]
& S \arrow[r, "b", no head, shift left]
\arrow[r, no head]
\arrow[r, no head, shift right]
& B \\
\end{tikzcd}
\end{equation}
with a single fact between $A$ and $S$. Then we have, for any $b \in F_{SB}$,
\begin{equation}
    P_{{S}} (a_0|b) = \frac{|W_{{S}} (a_0,b)|^2}{|W_{{S}} (a_0,b)|^2} =  1,
\end{equation}
while $P(b|a_0)$ can take any real value between 0 and 1:
\begin{equation}
    P_{{S}} (b|a_0) = \frac{|W_{{S}} (b,a_0)|^2}{\sum_{b'} |W_{{S}} (b',a_0)|^2}.
\end{equation}
This example shows that, when there is a single fact within a fact-set, the fact is certain, while otherwise facts are only potentialities.

\passage{Incompatibility of parallel facts}

Consider the case when $W_{{S}}(a,b) = 0$, for two facts $a,b \in F_S$. Then
\begin{equation}
    P_{{S}} (a|b) = P_{{S}} (b|a) = 0.
\end{equation}
Conversely, the amplitude vanishes if the conditional probabilities do. In such a case, $a$ and $b$ are mutually excluding each other: we say the two facts are \textit{incompatible}. We expect that this can never happen if $a=b$, in which case we shall have instead
\begin{equation}
\label{eq:P=1}
    P_{{S}} (a|a) = 1.
\end{equation}

Two different facts within the same fact-set are said to be \textit{parallel}. Motivated by the examples above, we will assume that parallel facts are always incompatible. This implies that for two different facts $a_0 ,a_1 \in F_{SA}$ we have
\begin{equation}
    W_S (a_0,a_1)= 0.
\end{equation}
This condition suffices to guarantee \eqref{eq:P=1}. However, for simplicity, since this is the case in all the considered examples, we will assume a slightly stronger condition
\begin{equation}\label{eq:ind_par_facts}
    W_S (a_i,a_j) = \delta_{ij}.
\end{equation}
This simply corresponds to the normalisation of states that represent facts, as we will see in the coming sections.

\passage{Chain property}

We now discuss a property that is not required for a generic fact-net, but appears in numerous cases with interesting consequences. Let's consider an environment made of three systems:

\begin{equation}
    \begin{tikzcd}
A \arrow[r, "a", no head, shift left]
\arrow[r, no head] \arrow[r, no head, shift right]
& S \arrow[r, "b", no head, shift left] \arrow[r, no head]  \arrow[r, no head, shift right]
& B \\
& C \arrow[u, "c", no head, shift left] \arrow[u, no head] \arrow[u, no head, shift right] & 
\end{tikzcd}
\end{equation}
We say that the \textit{chain property} holds for $a \in F_{AS}$ and $b \in F_{BS}$ with respect to $C$ if
\begin{equation}\label{eq:chain_prop}
     W_S (a,b) = \sum_{c \in F_{CS}} W_S (a,c) W_S (c,b).
\end{equation}
If the property holds for \emph{all} $a \in F_{AS}$ and $b \in F_{BS}$, we say that it holds for $A$ and $B$ with respect to $C$. If the property holds for any such triple in $N_S $, we say that $S$ is \emph{chain-complete}.

This is a strong condition that implies powerful features of the theory. It has already been identified as an important property of quantum theory in \cite{rovelli2021a}.

\passage{Example 3: three directions of spins}

As an example, consider a spin-$\frac{1}{2}$ particle with three surrounding Stern-Gerlach devices:
\begin{equation}
    \begin{tikzcd}
\mathfrak{S}_z \arrow[r, "0", no head, shift left] \arrow[r, "1"', no head, shift right] 
& S \arrow[r, "+", no head, shift left]  \arrow[r, "-"', no head, shift right]
& \mathfrak{S}_x \\
&   \mathfrak{S}_y \arrow[u, "i", no head] \arrow[u, "-i"', no head, shift right] & 
\end{tikzcd}
\end{equation}
Equation \eqref{eq:amplitude_2SG} already gives the amplitude between facts in $F_{\mathfrak{S}_z S}$ and $F_{\mathfrak{S}_xS}$. We must complete the definition of $W$ to account for the facts in $F_{\mathfrak{S}_y S}$:
\begin{equation}
\label{eq:W_yz}
    \begin{split}
        &W_S(0, \pm i) = \frac{1}{\sqrt{2}} \\
        &W_S(1, \pm i)= \pm \frac{i}{\sqrt{2}}
    \end{split}
\end{equation}
and
\begin{equation}
\label{eq:W_yx}
    \begin{split}
        &W_S(+, \pm i) = \frac{1 \pm i}{2} \\
        &W_S(-, \pm i)=  \frac{1- \pm i }{2}
    \end{split}
\end{equation}
This gives the right probability amplitudes between different possible outcomes of the relevant experiments. One can then check that such system $S$ is chain-complete. In fact, by assuming that $S$ is chain-complete, one can compute \eqref{eq:W_yx} from \eqref{eq:amplitude_2SG} and \eqref{eq:W_yz}.

\passage{Example 4: propagator}

Let's provide another example of a chain-complete system. Consider a free particle of mass $m$ moving in a $1$-dimensional space and a position measurement device $X$ turned on three times $t_0,t_1,t_2$. In this setup, time is an external parameter and $X$ generates three different fact-sets, each labelled by a different moment in time. So, on the fact-net, $X$ appears as three different systems:
\begin{equation}\label{eq:propagator_3}
    \begin{tikzcd}
X_0 \arrow[r, "x_0", no head, shift left]
\arrow[r, no head] \arrow[r, no head, shift right] 
& S \arrow[r, "x_1", no head, shift left]  \arrow[r, no head] \arrow[r, no head, shift right]
& X_1 \\
& X_2 \arrow[u, "x_2", no head, shift left] \arrow[u, no head] \arrow[u, no head, shift right] & 
\end{tikzcd}
\end{equation}
Then the amplitude is
\begin{equation}\label{eq:propagator_free}
    W_S(x_i , x_j) = \left( \frac{1}{2 \pi (i \hbar (t_i - t_j) /m + \sigma^2 )} \right)^{1/2}e^{- \frac{ (x_i-x_j)^2}{2 (i \hbar (t_i - t_j)/ m + \sigma^2) }}
\end{equation}
where $\sigma$ is the precision of $X$, i.e. the smallest distance that $X$ can resolve.
This is the standard formula for the propagator of a free particle. One checks easily that $W_S$ is hermitian. The incompatibility of parallel facts as expressed by eq. \eqref{eq:ind_par_facts} is taken as a definition of the amplitude within the same fact-set. With some more work, one can check that $S$ is chain-complete. Interestingly, no ordering between $t_0, t_1$ and $t_2$ is required for these properties to hold. Besides, one can compute the conditional probability density
\begin{equation}
\label{eq:probability_propagator}
    P_{{S}} (x_i | x_j) = \left(2 \pi \left( \frac{\hbar^2 (t_i - t_j)^2}{ m^2 \sigma^2} + \sigma^2 \right) \right)^{-\frac{1}{2}} \, e^{- \frac{ (x_i-x_j)^2 }{ \frac{\hbar^2 (t_i - t_j)^2}{ m^2 \sigma^2} + \sigma^2 }}
\end{equation}
We see that a non-zero $\sigma$ is necessary to get a non-trivial probability and indeed, in experiments, the precision on $X$ is never perfect. So the fact-sets $F_{SX}$ is not really $\mathbb{R}$, as we have assumed to perform the integration to get \eqref{eq:probability_propagator}, but rather a discrete set of possible intervals of size $\sigma$.

In standard quantum mechanics, the time evolution of an isolated system is described as a unitary transformation of states within the same Hilbert space. In the fact-net framework, this evolution appears as a relation between sets of potential facts associated with measurement apparatuses at different moments in time. So there is no such thing as the time evolution of a fact-net because time must be already included within the fact-net.

The description of time evolution illustrates a slight departure from the more common use of the word `system' by physicists. The traditional denomination implicitly assumes the existence of an operational labelling \textit{t} and identifies different snapshots
as the same system. Our convention enlarges the use of the word ‘system’ to refer to the snapshots themselves, which are believed to be ontologically primary. It is mainly a matter of word, but we believe our convention is more in line with a generally covariant philosophy.

As we will see in section \ref{sec:amplitude_maps}, it is still possible to express $W_S$ as a unitary map between the facts at different moments of time, looking like
\begin{equation}\label{eq:amp_map1}
        \sum_{x_0,x_1} W_S(x_0,x_1) \dyad{x_0}{x_1}.
\end{equation}
However, one should bear in mind the important conceptual difference with standard quantum mechanics: we are not assuming the existence of an intrinsic evolution of the state in between the measurements (the facts).

\passage{Example 5: $\beta$-decay}

For the reason just mentioned, the fact-net formalism is also more suitable to express the time-reversal symmetry of quantum mechanics. Indeed, in standard quantum mechanics, the quantum state attributed to a system in-between two measurements is not time-reversal symmetric. This was shown in \cite{rovelli2016} with a striking example:
\begin{quote}
\textit{We can describe the $\beta$-decay of a nucleus by means of the electron wave function concentrated on the nucleus, slowly leaking out in all directions until the electron gets detected by a Geiger counter at some distance. The time reversed phenomenon –the Geiger apparatus emits an electron that is then captured by the nucleus– is perfectly possible, but it is not described by a wave function that converges symmetrically onto the nucleus.}
\end{quote}
This argument was raised against a realistic interpretation of the quantum state. The fact-net corresponding to this example is 
\begin{equation}
   \begin{tikzcd}
Nucleus \arrow[r, "1"', no head, shift right]  \arrow[r, "0", no head, shift left]
& Electron \arrow[r, no head, shift right] \arrow[r, no head] \arrow[r, "x", no head, shift left]
& Geiger \, counter 
\end{tikzcd} 
\end{equation}
The facts 0 and 1 tell whether the nucleus has decayed or not, so for instance, whether the nucleus is carbon-14 or nitrogen-14. The facts $x$ correspond to a click of the Geiger counter in position $x$. In standard quantum mechanics, knowing that the nucleus is initially carbon-14 allows describing the evolution of the state of the electron at all times later, until the electron is detected somewhere outside, by the Geiger counter. However, the backward evolution of this state disagrees with what standard quantum mechanics would predict for the evolution of the state of an electron that would be emitted by the Geiger counter. In that sense, standard quantum mechanics is not time-reversal symmetric. The fact-net formalism carries an easy solution to the difficulty: there is no continuity of the evolution to be assumed in-between a pair of facts.

\passage{Example 6: entanglement}

In the examples above, the central system $S$ was thought of as the system observed by the surrounding measurement apparatuses. However, the formalism of fact-nets is symmetric and thus admits a mirror interpretation where $S$ plays the role of an observer, observing its neighbour systems. Consider the following fact-net
\begin{equation}
    \label{eq:entanglement-fact-net}
   \begin{tikzcd}
A \arrow[r, "0"', no head, shift right]  \arrow[r, "1", no head, shift left]
& O \arrow[r, "0"', no head, shift right]  \arrow[r, "1", no head, shift left]
& B 
\end{tikzcd} 
\end{equation}
where $A$ and $B$ are qubits and $O$ the observing system. Then the amplitude 
\begin{equation}
\label{eq:entanglement}
    W_{{O}}(i_A,j_B) = \frac{\delta_{i_A,j_B}}{\sqrt{2}},
\end{equation}
for $i_A,j_B \in \{0,1\}$, describes an entangled state of the coupled system $A \cup B$. It corresponds to the Bell state
\begin{equation}
   \frac{1}{\sqrt{2}} \left(\ket{00} + \ket{11} \right),
\end{equation}
and should be understood as a state of $A\cup B$ \emph{relative to $O$}, or, equivalently, a state of $O$ relative to $A\cup B$. The description of composite systems in the fact-net formalism is a subject of section \ref{sec:composite_systems}.

\passage{Example 7: spinor}

In the example of the $\beta$-decay, we have seen a fact-net where the fact-sets don't have the same cardinality. Let's see another such example: a particle for which one can measure both the position and the spin. The fact-net is 
\begin{equation}
    \begin{tikzcd}
\mathfrak{S}_z \arrow[r, "0", no head, shift left] \arrow[r, "1"', no head, shift right] 
& S \arrow[r, "x", no head, shift left]
\arrow[r, no head]
\arrow[r, no head, shift right]
& X \\
\end{tikzcd}
\end{equation}
We use the variable $\sigma$ for the facts of spins and the variable $x$ for the facts of position. Then the amplitude is a function $W(x,\sigma)$, that is usually called a spinor. It expresses a correlation between the spin and the position. In the standard formalism, this would correspond to an entangled state between spin and position variables, like
\begin{equation}\label{eq:amp_state1}
        \sum_{x , \sigma }W_{{S}} (x,\sigma) \ket{x}\ket{\sigma}.
\end{equation}

If there are only two facts $x_0$ and $x_1$ available between $S$ and $X$ we can model a particle going through a Stern-Gerlach by the amplitude, for $i \in \{0,1\}$
\begin{equation}
    W_S (i,x_j) = \frac{\delta_{ij}}{\sqrt{2}},
\end{equation}
which is analogous to \eqref{eq:entanglement}. Indeed, the Stern-Gerlach creates a correlation between a spin orientation and a position.

\passage{What is the amplitude?}

From the different examples above, we see that the interpretation of W is very different depending on the physical situation that the fact-net is meant to represent. There are two main ways to look at it:
\begin{itemize}
    \item A map that performs a change of basis or a unitary time evolution. 
    \item A state of a bipartite system, which covers cases like a spinor (one system with several degrees of freedom) or Bell's states (two systems with correlated d.o.f).
\end{itemize}
The duality between these views is provided by the usual Choi-Jamiołkowski isomorphism. 

This dual interpretation of the amplitude adds to another form of duality in the interpretation of the fact-nets. Indeed, in all cases encountered so far, the fact-net was \textit{star-shaped}, with one central system surrounded by a few others, which are not related to each other by any facts. There are two main ways to consider such a fact-net. Either the central system is regarded as a quantum system and the surrounding ones as measurement apparatuses, or the central system is thought of as a classical observer and the surrounding systems are a collection of quantum observed systems. It is noticeable that these two symmetric cases are treated identically in the fact-net formalism. This is because fact-nets do not assume any preferential interpretations as observed or observer: both perspectives are allowed. In section \ref{sec:QRF}, we will encounter different types of fact-nets, with triangles of systems, which are suitable to describe changes of the quantum reference frames.

\section{Composite systems}
\label{sec:composite_systems}

The union of two systems is again a system. Consider the fact-net
\begin{equation}\label{eq:composite_two}
   \begin{tikzcd}
A \arrow[r, no head] \arrow[r, no head, shift left]
& S \arrow[r, no head] \arrow[r, no head, shift left] \arrow[r, no head, shift right]
& B
\end{tikzcd} 
\end{equation}
Another description could be given by a fact-net of the following form
\begin{equation}\label{eq:composite_one}
   \begin{tikzcd}
S \arrow[r, no head, shift right] \arrow[r, no head, shift right, shift right] \arrow[r, no head, shift right, shift right] \arrow[r, no head] \arrow[r, no head, shift left] \arrow[r, no head, shift left, shift left]  \arrow[r, no head, shift left, shift left, shift left]
& A \cup B
\end{tikzcd},
\end{equation}
where the systems $A$ and $B$ are composed into a single system, and we consider facts to be pairs of facts from the original diagram
\begin{equation}\label{eq:composite_facts}
    F_{S, A \cup B} = F_{S,A} \times F_{S,B}.
\end{equation}
However, from \eqref{eq:composite_two} to \eqref{eq:composite_one}, some information is lost because the non-trivial amplitude $W(a,b)$ only appears in \eqref{eq:composite_two}. 
Both situations can also fit on the same fact-net as
\begin{equation}\label{diag:composite}
   \begin{tikzcd}
A \arrow[r, no head, shift left] \arrow[r, no head]
& S \arrow[r, no head] \arrow[r, no head, shift left] \arrow[r, no head, shift right]
& B \\
 & A \cup B \arrow[u, no head, shift left] \arrow[u, no head, shift left, shift left] \arrow[u, no head] \arrow[u, no head, shift right] \arrow[u, no head, shift right, shift right] \arrow[u, no head, shift right, shift right, shift right]  & 
\end{tikzcd} 
\end{equation}
What should be the amplitude between the facts $(a_0,b) \in F_{S, A \cup B}$ and  $a_1 \in F_{SA}$? We expect the probability of getting $(a_0,b)$ conditioned on $a_1$ to be zero if $a_0 \neq a_1$ and $P_{{S}} (b|a_0)$ otherwise, i.e.
\begin{equation}\label{eq:composite_proba}
    P_{{S}} ((a_0,b)|a_1) = \delta_{a_0,a_1} P_{{S}} (b|a_1).
\end{equation}
For the amplitude, it implies \begin{equation}
    W_{{S}} ((a_0,b),a_1 ) \propto \delta_{a_0,a_1} W_{{S}} (b,a_1).
\end{equation}
The value of $W_{{S}} ((a,b),a)$ is not further constrained, but the simplest choice would clearly be
\begin{equation}\label{eq:composite_ampli}
    W_{{S}} ((a_0,b),a_1 ) = \delta_{a_0,a_1} W_{{S}} (b,a_1).
\end{equation}
We thus treat (\ref{eq:composite_ampli}) as completing the definition of the amplitude for the fact-net (\ref{diag:composite}). Notice here that for the system $S$ in (\ref{diag:composite}) the chain property does \emph{not} hold for any configuration of facts and intermediate systems.

So far, we have not said anything about the facts between a system and itself. We can now remark the following. If one maintains the rule \ref{eq:composite_facts} when $S=A$, a condition is imposed on the cardinality of $|F_{AA}|$. Indeed, $F_{AA} = \emptyset$ would imply $F_{A,A\cup B} = \emptyset$, which does not seem reasonable. However, one can take instead that $|F_{AA}|=1$, which would mean that each system has one fact with itself, the fact of being itself. Then it would imply $F_{A,A\cup B} = F_{AB} = F_{B,A\cup B}$, which can be understood as follows: a subsystem can only relate to the bigger system by relating to its complement.

\passage{Yes/no measurement}

Reciprocally, it is also possible to decompose a fact-net by splitting a system into subsystems. The maximal decomposition is done in terms of qubits corresponding to yes/no measurements. Consider two systems like
\begin{equation}
\label{eq:yes-no-AS}
   \begin{tikzcd}
A \arrow[r, no head, shift right] \arrow[r, no head] \arrow[r, "a_i", no head, shift left]
& S
\end{tikzcd} 
\end{equation}
Each fact $a_i \in F_{SA}$ can be promoted to an interacting qubit (yes/no measurement), and we get the equivalent fact-net:
\begin{equation}
\label{eq:splitted-yes-no}
   \begin{tikzcd}
 \dots \arrow[dr, no head] \arrow[dr, no head, shift left]   & a_i \arrow[d, no head] \arrow[d, no head, shift right] & \dots \arrow[dl, no head] \arrow[dl, no head, shift right] \\
a_1 \arrow[r, "0"', no head, shift right] \arrow[r, "1", no head]
& S \arrow[r, no head, shift right] \arrow[r, no head] & a_N
\end{tikzcd} 
\end{equation}
Because of the incompatibility of parallel facts in \eqref{eq:yes-no-AS}, we must have, for $i \neq j$,
\begin{equation}
\label{eq:W(1,1)}
    W_{{S}} (1_i,1_j) = 0.
\end{equation}
Conversely, it is possible to reduce a fact-net from \eqref{eq:splitted-yes-no} to \eqref{eq:yes-no-AS} provided \eqref{eq:W(1,1)} is satisfied. However, doing so, one loses the information contained in $W_{{S}} (0_i,0_j)$ and $W_{{S}} (0_i,1_j)$ which does not appear in \eqref{eq:yes-no-AS}. 

Concretely, one could think of a particle moving on a circle with a discrete but high number $N$ of possible sites $x_i$. In this case, one could assume the equiprobability of sites which leads to the conditional probabilities
\begin{equation}
\begin{split}
    &P_{{S}} (1_i | 0_j) = \frac{1- \delta_{a_i a_j}}{N-1} \\
    &P_{{S}} (0_i | 1_j) = 1 - \delta_{a_i a_j} \\
    &P_{{S}} (0_i | 0_j) =  \frac{N-2 + \delta_{a_i a_j}}{N-1} 
\end{split}
\end{equation}
In the limit of large $N$, these probabilities are satisfied by the amplitude
\begin{equation}
    \begin{split}
        &W_{{S}} (0_i, 0_j) = 1 \\
        &W_{{S}} (0_i,1_j) = 0.
    \end{split}
\end{equation}

\section{Measurement}
\label{sec:measurement}

A fact-net provides the structural relations between a set of interacting systems. The different facts within the same fact-set are different possibilities for the result of the interaction. So a fact-net gathers many potential configurations of facts.

During a measurement procedure, only one fact in the fact-set is actualized. When one fact is chosen for each fact-set, the restricted fact-net can be seen as a \textit{branch} of the world, to mimic the language of many-worlds. More generally, one can consider restrictions of the fact-sets, where not just a single fact is selected, but a full subset of a fact-set. This would correspond to a post-selection of the possible facts. It results in a new fact-net with a new amplitude $W'$ as we now explain.

A \textit{measurement relative to a reference system O} is a process $M_O$ that sends a fact-net to another fact-net. The new fact-net has the same systems, but the fact-sets are restricted to a subset $M_{OA} \subset F_{OA}$ for all $A \in N_O$.

After the restriction, the amplitude $W$ is still hermitian and parallel facts are still incompatible. However, the chain property is generally broken. Yet, it is always possible to restore the chain property for a chosen ordering of $N_O$ by defining a new amplitude $W'$.

First, choosing an ordering of $N_O$ amounts to introduce a local and primitive notion of time: the facts happen in successive order. Let's then exemplify the construction of $W'$ for the following fact-net with an ordering of $N_S$ going from A to D.
\begin{equation}
   \begin{tikzcd}
B \arrow[dr, no head] \arrow[dr, no head, shift right]   &  &  C \arrow[dl, no head] \arrow[dl, no head, shift right]\\
A \arrow[r, no head, shift right] \arrow[r, no head] 
& S \arrow[r, no head, shift right] \arrow[r, no head] & D
\end{tikzcd} 
\end{equation}
We first define $W'$ for closest neighbours in the ordering of $N_O$, i.e.
\begin{equation}
\begin{split}
    &W_{{S}} '(a,b) \overset{\text{def}}= W_{{S}} (a,b) \\
    &W_{{S}} '(b,c) \overset{\text{def}}= W_{{S}} (b,c) \\
    &W_{{S}} '(c,d) \overset{\text{def}}= W_{{S}} (c,d)
\end{split}
\end{equation}
Then, the other $W'$ are defined by imposing the chain property on intermediate facts:
\begin{equation}
\begin{split}
    &W_{{S}} '(a,c) \overset{\text{def}}= \sum_{b \in R_{SB}} W_{{S}} (a,b)W_{{S}} (b,c). \\
    &W_{{S}} '(b,d) \overset{\text{def}}= \sum_{c \in R_{SC}} W_{{S}} (b,c)W_{{S}} (c,d) \\
    &W_{{S}} '(a,d) \overset{\text{def}}= \sum_{b \in R_{SB}} \sum_{c \in R_{SC}} W_{{S}} (a,b)W_{{S}} (b,c)W_{{S}} (c,d).
\end{split}
\end{equation}
This construction generalizes easily for an arbitrary number of neighbours. For an ordering of them like $X_0 \to X_1 \to ... \to X_n$, the chain property is satisfied between any $X_i$ and $X_j$ with respect to any $X_k$ with $i<k<j$.

\passage{Double-slit experiment}

We can apply the previous construction to understand the double-slit experiment. The latter can be constructed as a post-selection from the fact-net of the propagator \eqref{eq:propagator_3}. Usually, one proceeds the other way around and starts from the double-slit experiment, and then progressively increases the number of slits, until the slits give way to free space. This motivates the Feynman path-integral, which enables the computation of the propagator. Here, we illustrate the measurement process on fact-nets by considering that the slits are a position measurement device that only keeps two facts in the fact-set. Similarly, the source measures the position of the electron with only one possible fact. So the situation is the following:
\begin{equation}
   \begin{tikzcd}
   & Slits \arrow[d, "A"', no head, shift right] \arrow[d, "B", no head]   &  \\
Source \arrow[r, "1", no head]
& Electron \arrow[r, no head, shift right] \arrow[r, no head] \arrow[r, "x", no head, shift left]
& Screen
\end{tikzcd}
\end{equation}
We assume that the electron first leaves the source, then goes through the slits $A$ and $B$ located at a distance $\ell_1$ from the source and separated by a distance $2d$, and finally reaches the screen at a distance $\ell_2$ from the slits. This scenario carries a local ordering of the facts for the electron: source, slits, screen. The amplitude can be deduced from that of the propagator, eq. \eqref{eq:propagator_free}. Between the source and the slits, and the slits and the screen, the formula only requires an easy adaptation to transform the difference of times $(t_i-t_j)$ into distances $\ell$:
\begin{equation}
\begin{split}
     &W_{{E}} '(A,1) = \left( \frac{1}{2 \pi (i \ell_1^2 + \sigma^2 )} \right)^{1/2}e^{- \frac{ d^2}{2 (i \ell_1^2 + \sigma^2) }} \\
     &W_{{E}} '(x,A) = \left( \frac{1}{2 \pi (i \ell_2^2 + \sigma^2 )} \right)^{1/2} e^{- \frac{ (x-d)^2}{2 (i \ell_2^2 + \sigma^2) }}
\end{split}
\end{equation}
Then the amplitude between the source and the screen must be so that the chain property is satisfied, which is necessary to get the right prediction of the interference pattern. So we define
\begin{equation}
    W_{{E}} '(x,1) \overset{\text{def}}= W_{{E}} '(x,A)W_{{E}} '(A,1) + W_{{E}} '(x,B)W_{{E}} '(B,1).
\end{equation}
With the amplitude defined as such, the chain property is only satisfied between the source and the screen with respect to the slits, i.e. for the ordering of $N_{Electron}$ which is implied by the local clock.

\section{Recovering the Hilbert space formalism}\label{sec:recovering}

One of the goals for developing the fact-net formalism is to kick Hilbert spaces out of quantum mechanics. However, it is important to show that we are not losing anything crucial. In this section, we show that fact-nets are sufficient to recover the standard description of quantum physics.

\passage{Relative Hilbert spaces}

Consider the fact-net
\begin{equation}
   \begin{tikzcd}
A \arrow[r, no head, shift right] \arrow[r, no head] \arrow[r, no head, shift left]
& S
\end{tikzcd} 
\end{equation}
In the spirit of RQM, we first define the Hilbert space of $S$ \emph{relative} to $A$ as
\begin{equation}
    \mathcal{H}_{S|A} \overset{\text{def}}= l^2(F_{AS}).
\end{equation}
It is the complex vector space freely generated by the facts in $F_{AS}$, i.e. any vector reads:
\begin{equation}
    \sum_{a \in F_{AS}} c_a \ket{a},
\end{equation}
with $c_a \in \mathbb{C}$ so that 
\begin{equation}
    \sum_{a \in F_{AS}} |c_a|^2 < \infty.
\end{equation}
The scalar product is defined so that the basis $\{ \ket{a} \mid a \in F_{AS} \}$ is orthonormal, so that $W(a_i,a_j) = \delta_{ij} = \braket{a_i}{a_j}$.

The notion of relative Hilbert space introduces in the language an asymmetry between the observed system and the reference system (the observer). However, this distinction is not present at the level of the fact-net, and for this reason, we have $\mathcal{H}_{S|A} = \mathcal{H}_{A|S}$.

It should be noted that, in this framework, the Hilbert space $\mathcal{H}_{S|A}$ of $S$ relative to $A$ comes with a preferred basis. Indeed, the set of possible facts between two systems is not derived from an \textit{a priori} definition of the systems, but it is instead part of the \textit{definition} of what the interaction between two systems is. In other words, systems are characterized by the way they can connect to their environment. This feature of the formalism is simply the translation of: 1) in relational quantum mechanics interactions are ontologically of the same nature than measurements, and 2) measurements come with a preferred basis.

We are now going to add other systems in $N_S$ and express them as operators on $ \mathcal{H}_{S|A}$. Consider a system $B \in N_S$:
\begin{equation}
   \begin{tikzcd}
A \arrow[r, no head] \arrow[r, no head, shift left] \arrow[r, no head, shift right]
& S \arrow[r, no head, shift right] \arrow[r, no head] \arrow[r, no head, shift left]
& B 
\end{tikzcd}.
\end{equation}
Notice, first, that a consequence of \eqref{eq:composite_facts} is that
\begin{equation}
    \mathcal{H}_{S|A \cup B} \cong \mathcal{H}_{S|A} \otimes \mathcal{H}_{S|B},
\end{equation}
so that a relative Hilbert space of a composite system is given by the tensor product, as expected. Generally speaking, one could consider the \emph{extended} Hilbert space of $S$:
\begin{equation}
    \mathcal{H}_S^{ext}\overset{\text{def}}= \bigotimes_{A \in N_S} \mathcal{H}_{S|A}
\end{equation}
It is the Hilbert space of $S$ \emph{relative to its whole environment}.

Then, the amplitude $W_S $ enables to express a fact $b \in F_{SB}$ as a vector in $\mathcal{H}_{S|A}$
\begin{equation} \label{ketb}
    \ket{b} \overset{\text{def}}= \sum_{a \in F_{AS}} W_S(a,b) \ket{a}.
\end{equation}
The system $B$ can then be expressed as a \textit{relative observable}, i.e. a self-adjoint operator $\hat B \in \mathcal{L}(\mathcal{H}_{S|A})$, defined by the spectral decomposition:
\begin{equation}\label{eq:obs}
    \hat B \overset{\text{def}}= \sum_{b \in F_{SB}} f(b) \dyad{b},
\end{equation}
with a function $f: F_{SB} \to \mathbb{R}$. In our framework, the facts $b$ are elements of an abstract set $F_{SB}$ and this explains the need to introduce the function $f$ to turn these facts into numerical outcomes of measurements, so that the fact-sets become sample spaces. The choice of the function $f$ is conventional and includes specifying a coordinate or a unit system. It is of course possible to label the fact-set directly by the eigenvalues that one expects for $B$ so that $f$ is the identity.

Now consider a third system $C \in N_S$
\begin{equation}
    \begin{tikzcd}
A \arrow[r, "a", no head, shift left]
\arrow[r, no head] \arrow[r, no head, shift right]
& S \arrow[r, "b", no head, shift left] \arrow[r, no head]  \arrow[r, no head, shift right]
& B \\
& C \arrow[u, "c", no head, shift left] \arrow[u, no head] \arrow[u, no head, shift right] & 
\end{tikzcd}
\end{equation}
Then, it is easy to show that two facts $b \in F_{SB}$ and $c \in F_{SC}$ satisfy the chain property with respect to $A$ if and only if 
\begin{equation}
    \label{eq:braket=W}
    \braket{b}{c} = W_S(b,c).
\end{equation}
If $S$ is chain-complete, then this equality is always true, so that all the facts involving $S$ are expressed as states in $\mathcal{H}_{S|A}$ and $W_S$ gives the value of their scalar product.

We have thus recovered the main features of the standard formulation of quantum mechanics. The salient point of this construction is the need to take one system (here $A$) as a reference system, with respect to which the other facts are described. This departs from the fact-net formalism, where all systems are treated on the same ground.

\medskip

As an example, consider a particle $S$, a ruler $X$ and an apparatus $P$ that measures momentum. The fact-net is
\begin{equation}
   \begin{tikzcd}
X \arrow[r, no head, shift right] \arrow[r, no head] \arrow[r, "x", no head, shift left]
& S \arrow[r, no head, shift right] \arrow[r, no head] \arrow[r, "p", no head, shift left]
& P 
\end{tikzcd} 
\end{equation}
The fact-sets are indexed by real variables $x,p \in \mathbb{R}$. The amplitude is
\begin{equation}
    W_{{S}} (x,p) =  e^{\frac{i}{\hbar} p x}
\end{equation}
The Hilbert space of $S$ relative to $X$, denoted $\mathcal{H}_{S|X}$, is spanned by the $\ket{x}$. Then, the facts $p$ can be represented as states in $\mathcal{H}_{S|X}$,  
\begin{equation}
    \ket{p} \overset{\text{def}}=  \sum_x e^{\frac{i}{\hbar} p x} \ket{x}
\end{equation}
$X$ and $P$ can be represented as observables on $\mathcal{H}_{S|X}$,
\begin{equation}
    \begin{split}
        \hat X &= \sum_x x \dyad{x} \quad,\quad \sum_x \overset{\text{def}}= \int dx \\
        \hat P &= \sum_p p \dyad{p} \quad,\quad \sum_p \overset{\text{def}}= \int \frac{dp}{2\pi\hbar} 
    \end{split}
\end{equation}
Here, the function $f$ is the identity because the fact-sets are already labelled by the value of the measurement outcome. 
The momentum operator then reads
\begin{align}
    \hat P &=  \sum_{x,x',p} p\, e^{\frac{i}{\hbar} p (x-x')} \dyad{x}{x'} \nonumber\\
    &=  \sum_{x,x'} (-i\hbar)\partial_x ( \delta (x-x') ) \dyad{x}{x'} ,
\end{align}
and we recover the canonical commutation relations
\begin{equation}
\begin{split}
       \hat{X}\hat{P} - \hat{P}\hat{X}
    &= \sum_{x,x'} (-i\hbar) (x-x') \partial_x( \delta (x-x') ) \dyad{x}{x'} \\
    &= \sum_{x,x'} i\hbar \, \delta (x-x') \dyad{x}{x'} \\
    &= i\hbar\, \mathbb{1}_{\mathcal{H}_{S|X}} .
\end{split}
\end{equation}

\section{Self-space of a system}
\label{sec:representing}

Let's consider a system $S$ in a fact-net surrounded by systems $A_i$. We have already introduced the relative Hilbert spaces $\mathcal{H}_{S|A_i}$ describing $S$ from the point of view of $A_i$. The communication between those relative Hilbert spaces is governed by the $W$ maps. One can wonder if it is possible to assign to each system a Hilbert space that is not relative to any other system in particular but is wide enough to account for the points of view of all the other systems, in the same way that the Hilbert space of the usual framework of quantum mechanics can account for the points of view of different measuring devices.

In this section, we will investigate under which condition we can assign such a \textit{self-space} $\mathcal{H}_S$ to $S$. In this self-space, we want all facts $f,g\in F_S$ to be represented by normalized vectors $\ket{f}_S,\ket{g}_S\in\mathcal{H_S}$ such that $\braket{f}{g}_S=W_S (f,g)$. We already know that this is the case for the relative Hilbert space $\mathcal{H}_{S|A}$, provided  $S$ is chain-complete. Our task is now to construct such a Hilbert space, called the self-space $\mathcal{H}_S$, so that we always have $W_S(b,c)=\braket{b}{c}$, even if the chain property is not satisfied. 

Let's start by introducing the Hilbert space:
\begin{equation}
    \mathcal{H}^{\oplus}_S \overset{\text{def}}= l^2 (F_S),
\end{equation}
with the scalar product defined such that $\{\ket{f} \mid f \in F_S\}$ forms an orthonormal basis. Note that $\mathcal{H}^{\oplus}_S= \bigoplus_i \mathcal{H}_{S|A_i}$. We can then define a linear endomorphism $\mathcal{W}^S : \mathcal{H}^{\oplus}_S \to \mathcal{H}^{\oplus}_S$ as:
\begin{equation}
    \mathcal{W}^S \ket{g}_S = \sum_{f \in F_S} W_S (f,g) \ket{f}_S.
\end{equation}
The following Lemma holds:

\begin{lemma}[Self-space of a system]\label{lem:rep}
    A finite system $S$ admits a self-space if and only if $\mathcal{W}^S $ is positive semi-definite.
\end{lemma}

\begin{proof}

Assuming that $\mathcal{H}_S $ exists, consider the linear map $H: \mathcal{H}^{\oplus}_S \to \mathcal{H}_S $ defined as $\ket{f} \mapsto \ket{f}_S $. We can compute:
\begin{equation}
    \bra{f}\mathcal{W}^S \ket{g} = W_S (f,g) = \braket{f}{g}_S = \bra{f} H^\dagger H \ket{g},
\end{equation}
where the first equality holds by definition, the second is the assumption on $\mathcal{H}_S$, and the third is the definition of $H$. So we have $\mathcal{W}^S = H^\dagger H$ which means that $\mathcal{W}^S$ is positive semi-definite. In other words, for all vectors $\ket{x}$, $\bra{x} \mathcal{W}^S \ket{x}\geq 0$.

Conversely, assuming that $\mathcal{W}_S$ is hermitian semi-definite positive, it can be unitarily diagonalized into a diagonal matrix $D$ with non-negative coefficients. We then have $\mathcal{W}^S=U^\dagger D U$ where $U$ is a unitary matrix on $\mathcal{H}^{\oplus}_S$. We define $\mathcal{H}_S \overset{\text{def}}= \text{Im}(D)$ and an isometry $P: \text{Im}(D) \to \mathcal{H}^{\oplus}_S $ such that $D'=P^\dagger D P$ is a full rank diagonal matrix with non-negative coefficients and $PD' P^\dagger =D$. Then we can write $D'=\Delta^\dagger \Delta$ with $\Delta \overset{\text{def}}= \sqrt{D'}$. Finally, we have $\mathcal{W}_S =U^\dagger P\Delta^\dagger \Delta P^\dagger U $. Setting $H\overset{\text{def}}= \Delta P^\dagger U$ provides the desired decomposition, i.e. a surjective map $H:\mathcal{H}^{\oplus}_S \to \mathcal{H}_S$ such that for $\ket{f}_S \overset{\text{def}}= H \ket{f}$ we have:
\begin{equation}
    \braket{f}{g}_S =  \bra{f} H^\dagger H \ket{g} = \bra{f} \mathcal{W}_S \ket{g} = W_S (f,g),
\end{equation}
and since the equality holds for arbitrary facts $f$ and $g$, the constructed $\mathcal{H}_S$ is a self-space for $S$. 

\end{proof}

Note that this construction depends on the chosen decomposition $H^\dagger H$ for $\mathcal{W}_S $, which is unique only up to a unitary since $H^{\dagger} U^{\dagger} U H = H^\dagger H$. Hence, the self-space $\Im{H}$ is defined up to unitary transformations, unlike the relative Hilbert spaces that come equipped with preferred bases. Note that the dimension of $\mathcal{H}_S $ is the rank of $\mathcal{W}^S$. Moreover, our convention about parallel facts implies that $\mathcal{W}^S$ has identity diagonal blocks of size $|F_{SA_i }|$. It follows that the dimension of $\mathcal{H}_S $ is bounded by $\min_i |F_{A_i } |$ below and $|F_S |= \sum_i |F_{A_i } |$ above. One can easily embed the relative Hilbert spaces as subspaces of $\mathcal{H}_S $ by restricting the map $H$.

This construction is in line with the relational idea that a system is nothing more than the potentialities of interactions with its environment. Hence, the self-space represents the internal degrees of freedom of the system as emerging from the facts involving it. Let's move to concrete examples.

For the very simple fact-net 
\begin{equation}
   \begin{tikzcd}
A \arrow[r, no head, shift right] \arrow[r, no head] \arrow[r, no head, shift left]
& B
\end{tikzcd} 
\end{equation}
we see that the self-space matches with the relative one
\begin{equation}
    \mathcal{H}_A \cong \mathcal{H}_{A|B}.
\end{equation}
Indeed, the matrix $\mathcal{W}_S $ is then the identity on $\mathcal{H}_S^\oplus$ so that the map $H$ can be any unitary.

As a less simple illustration, consider the setup of two Stern-Gerlach devices as depicted in (\ref{2SG}). We have four facts in $F_S$, so that $\mathcal{H}^{\oplus}_S = \mathbb{C}^4$. The matrix $\mathcal{W}^S$ in the $\{0,1,+,-\}$ basis reads
\begin{equation}
\mathcal{W}^S = 
\begin{bmatrix}
1 & 0 & \frac{1}{\sqrt{2}} & \frac{1}{\sqrt{2}} \\
0 & 1 &\frac{1}{\sqrt{2}} & -\frac{1}{\sqrt{2}} \\
\frac{1}{\sqrt{2}} & \frac{1}{\sqrt{2}} & 1 & 0 \\
\frac{1}{\sqrt{2}} & -\frac{1}{\sqrt{2}} & 0 & 1 \\
\end{bmatrix}
=
\begin{bmatrix}
\mathbb{1} & \mathbb{H} \\
\mathbb{H} & \mathbb{1}
\end{bmatrix} ,
\end{equation}
and upon diagonalization gives $\mathcal{W}^S =  UDU^{\dagger}$ with:
\begin{equation}
D = 2
\begin{bmatrix}
\mathbb{0} & \mathbb{0} \\
\mathbb{0} & \mathbb{1}
\end{bmatrix}
,\quad \textrm{and} \quad
U = \frac{1}{\sqrt{2}}
\begin{bmatrix}
\mathbb{H} & \mathbb{H} \\
-\mathbb{1} & \mathbb{1}
\end{bmatrix} .
\end{equation}
Clearly, the image of $D$ is two-dimensional, so the self-space is $H_S \cong \mathbb{C}^2$. In this simple case, it is isomorphic to both relative Hilbert spaces $\mathcal{H}_{S|\Sigma_x}$ and $\mathcal{H}_{S|\Sigma_z}$, which was to be expected for a chain-complete system.

The matrix $H^{\dagger}$ is here $\frac{1}{2}\begin{pmatrix}
\mathbb{H} & \mathbb{1}
\end{pmatrix}$.
The columns of $H^{\dagger}$ give us the decomposition of the vectors corresponding to the facts of $F_S $ in an orthonormal basis of $\mathcal{H}_S$. Notice that multiplying $H$ by a unitary on the left provides another acceptable representation of the facts in $F_S $.

\section{Amplitude maps}
\label{sec:amplitude_maps}

Consider the fact-net
\begin{equation}
   \begin{tikzcd}
A \arrow[r, no head, shift right] \arrow[r, no head] \arrow[r, no head, shift left]
& S \arrow[r, no head, shift right] \arrow[r, no head] \arrow[r, no head, shift left]
& B 
\end{tikzcd} 
\end{equation}
The amplitude $W_S$ enables to define a linear map, called the \emph{amplitude map}, $W^S_{AB} : \mathcal{H}_{S|B} \to \mathcal{H}_{S|A}$ given on the basis by
\begin{equation}\label{eq:amp_map2}
    W^S_{AB} \ket{b}_B = \sum_{a \in F_{SA}} W_S(a,b)\ket{a}_A.
\end{equation}
This representation of the amplitude provides a novel perspective on the chain property and can be seen as a tool for translating the relative observables between different relative Hilbert spaces. $W^S_{AB}$ is unitary when the chain property is satisfied between $A$ and $B$ with respect to themselves. For a chain-complete system, the amplitude maps are all unitary. As an example, for the pair of Stern-Gerlach devices, fact-net \eqref{2SG}, we have
\begin{equation}
    W^S_{\mathfrak{S}_x\mathfrak{S}_z} = 
\begin{bmatrix}
\frac{1}{\sqrt{2}} & \frac{1}{\sqrt{2}} \\
\frac{1}{\sqrt{2}} & -\frac{1}{\sqrt{2}} \\
\end{bmatrix}.
\end{equation}
For simplicity we now focus our attention on a star-shaped fact-net, with an arbitrary number of systems in $N_S$, although we only draw three below
\begin{equation}\label{diag:star}
\begin{tikzcd}
  & A \arrow[d, no head, shift right] \arrow[d, no head] \arrow[d, no head, shift left]  &   \\
  & S \arrow[ld, no head, shift left] \arrow[ld, no head] \arrow[rd, no head] \arrow[rd, no head, shift right] \arrow[rd, no head, shift right=2] \arrow[rd, no head, shift left] &   \\
B &       & C
\end{tikzcd}.
\end{equation}
The $W_S$ specifies the amplitude maps
\begin{equation}
   W^S_{AB}: \mathcal{H}_B \cong \mathcal{H}_{S|B} \to \mathcal{H}_{S|A} \cong \mathcal{H}_A 
\end{equation}
and so on. They can be all depicted on a single diagram in Hilbert spaces and linear maps
\begin{equation}\label{diag:amp_maps_triangle}
    \begin{tikzcd}
 & \mathcal{H}_A \arrow[rd, "W^S_{CA}"] &               \\
\mathcal{H}_B \arrow[ru, "W^S_{AB}"] \arrow[rr, "W^S_{CB}"'] &   & \mathcal{H}_C.
\end{tikzcd}
\end{equation}
This diagram fully represents the fact-net (\ref{diag:star}). In a way, the system $S$ of the star-shaped diagram is only auxiliary in the sense that all the relevant information -- \emph{the physics} -- is contained in the mutual relations, encoded by the amplitude maps, between the surrounding systems $A,B$ and $C$. The direction of the arrow is reversed by taking the hermitian conjugate since we have $W^A_{BC} = (W^A_{CB})^\dagger$.

\passage{Chain property as commutativity}

Finite chain-complete systems in star-shaped fact-nets are neatly characterized by the commutativity of the diagram of the amplitude maps, as the following Lemma shows.

\begin{lemma}[Chain-completeness as commutativity]\label{lem:com}
A finite system S in a star-shaped diagram is chain-complete if and only if all triangles composed from the amplitude maps are commutative. In such a case, all these maps are unitary.
\end{lemma}

\begin{proof}
We calculate $W^S_{CA} \circ W^S_{AB}$ and show that it equals $W^S_{CB}$ if and only if the chain property holds for $C$ and $B$ with respect to $A$. We have
\begin{equation}
    \begin{split}
     W^S_{AB} &= \sum_{a,b}W_S(b,a) \, {}_A\!\dyad{a}{b}_B,\\
     W^S_{CA} &= \sum_{c,a'}W_S(a',c) \, {}_C\!\dyad{c}{a'}_A,
\end{split}
\end{equation}
and we can calculate
\begin{align*}
   W^S_{CA} \circ W^S_{AB} &= \sum_{a,a',b,c} W_S(b,a)W_S(a',c){}_C\!\ket{c}\braket{a'}{a}\bra{b}_B \\
   &= \sum_{a,b,c} W_S(b,a)W_S(a,c){}_C\!\dyad{c}{b}_B \\
   &= \sum_{b,c} W_S(b,c){}_C\!\dyad{c}{b}_B = W^S_{CB},
\end{align*}
where we have first used orthogonality of the $\{\ket{a}\}$ basis, and then the chain property for $C$ and $B$ with respect to $A$. Clearly, the equality can be accomplished only if the third equality holds, i.e. this chain property can be used. Generalization to arbitrary number of systems surrounding $S$ is straightforward.

To see that the amplitude maps need to be unitary, take $B=C$ in the calculation above. We get
\begin{align}\label{eq:commute}
   (W^S_{AB})^\dagger \circ W^S_{AB} = \sum_{b,b'} W_S(b,b'){}_C\!\dyad{b'}{b}_B = Id_{\mathcal{H}_B}.
\end{align}

The other composition works the same way, the calculation can be done for any amplitude map. We can then conclude that they are all unitary maps.
\end{proof}

\passage{Translating relative observables}

In section \ref{sec:recovering}, we have shown that a system $B \in N_S$ can be expressed as a self-adjoint operator $\hat B$ on $\mathcal{H}_{S|A}$, where $A \in N_S$. To make explicit that $\hat B$ lives in $\mathcal{L}(\mathcal{H}_{S|A})$, let's add an index $A$, so that \eqref{eq:obs} becomes
\begin{equation}
    \hat{B}_A \overset{\text{def}}{=} \sum_{b \in F_{SB}} f(b) \, {}_A\!\dyad{b}_A.
\end{equation}
$\hat B_A$ is the observable $B$ from the perspective of $A$. Given another system $C \in N_S$, one can also consider $\hat B_C$, the observable $B$ from the perspective of $C$. Then, if the chain property is satisfied between $B$ and $C$ with respect to $A$, the amplitude map $W^S_{AC}$ enables to relate $\hat B_A$ and $\hat B_C$ as
\begin{equation}
\label{eq:BA_to_BC}
     \hat B_C = W^S_{CA} \hat B_A W^S_{AC}
\end{equation}

If $A$ and $C$ represent an observer at two different times, then $W^S_{AC}$ is unitary and equation \eqref{eq:BA_to_BC} shows the usual evolution of an observable $B$ in the Heisenberg picture.

Within the usual formulation of quantum mechanics, all the relative Hilbert spaces of a chain-complete system are identified and understood as ``the Hilbert space of the system''. For example, in a multiple Stern-Gerlach setup, ``the Hilbert space of $S$'' is said to be just $\mathbb{C}^2$, without any specific choice of basis or a reference system. In such a case all the relative Hilbert spaces are isomorphic to the self-space $\mathcal{H}_S$, which has no preferred basis, so we can say that the self-space $\mathcal{H}_S$ is what is usually considered as ``the Hilbert space of the system''. Then we consider different operators, all acting on $\mathcal{H}_S$, that correspond to different directions of the spin measurement in this simple example. Each of them can be diagonalized by a unitary change of basis.

In our framework, we can also work in the self-space of the system. Assuming chain-completeness, all the amplitude maps can be considered as endomorphisms of the self-space, and each basis of eigenvectors of a relative observable can be considered as a basis of $\mathcal{H}_S$, thanks to the inclusions $\phi_{SA}:\mathcal{H}_{S|A} \to \mathcal{H}_S$, which are now unitary maps. Under such identification, the amplitude maps are precisely the unitary transformations that one needs to perform to change the basis of $\mathcal{H}_S$ from that associated via relevant $\phi$ to the eigenvectors of, e.g., $\hat{A}_A$, to that of eigenvectors of $\hat{B}_B$. To see this, notice that $\ket{b}_A = W_{AB}\ket{b}_B$, which are now considered as vectors in $\mathcal{H}_S$. We then have
\begin{align*}
    \hat{B}_A &= \sum_{b \in F_{SB}} f(b)  {}_A\!\dyad{b}_A \\
    &= \sum_{b \in F_{SB}} f(b)  W_{AB}\ket{b}_{BB}\bra{b} W_{BA}\\
    &= W_{AB} \hat B_B W_{BA},
\end{align*}
which is a special case of the formula \eqref{eq:BA_to_BC}.

\passage{More commutative triangles}

Each relative Hilbert spaces $\mathcal{H}_{S|A}$ is included as a subspace of the self-space $\mathcal{H}_S$. This is captured by an \emph{isometry} $\phi_{SA}: \mathcal{H}_{S|A} \to \mathcal{H}_S$, such that
\begin{equation}
    \phi_{SA} \overset{\text{def}}{=} H\big|_{\mathcal{H}_{S|A}}: \ket{a}_A = \ket{a}_S,
\end{equation}
where the notation from the Lemma \ref{lem:rep} is used. The diagram (\ref{diag:amp_maps_triangle}) can then be completed to include also the self-space of the system $S$, so that the fact-net (\ref{diag:star}) is represented as a following diagram of Hilbert spaces and linear maps
\begin{equation}\label{startriangles}
\begin{tikzcd}
& \mathcal{H}_A \arrow[d, "\phi_{SA}"] \arrow[rdd, "W^S_{CA}", bend left] &        \\
 & \mathcal{H}_S   &    \\
\mathcal{H}_B \arrow[ru, "\phi_{SB}"] \arrow[ruu, "W^S_{AB}", bend left] &   & \mathcal{H}_C \arrow[lu, "\phi_{SC}"'] \arrow[ll, "W^S_{BC}"', bend left]
\end{tikzcd}.
\end{equation}

We always consider the maps in both directions, they come in pairs: a map and its hermitian conjugate. This way we get a bunch of commutative triangles. To see this, let us calculate
\begin{align*}
    (\phi_{SA})^\dagger \circ \phi_{SB} &= \sum_{a,b}{}_{A}\!\dyad{a}_S\dyad{b}_B\\
    &= \sum_{a,b} {}_{A}\!\ket{a} \braket{a}{b}_S \bra{b}_B\\
    &= \sum_{a,b} {}_{A}\!\ket{a} W_S(a,b)\bra{b}_B = W^S_{AB}.
\end{align*}
If $S$ is chain-complete, we also have
\begin{align*}
    \phi_{SA} \circ W^S_{AB} &= \sum_{a,a',b}W_S(a,b){}_S\!\dyad{a'}_A {}_A\dyad{a}{b}_B \\
    &= \sum_{a,b}W_S(a,b){}_S\!\dyad{a}{b}_B \\
    &= \sum_b {}_S\!\dyad{b}_B = \phi_{SB},
\end{align*}
It is compatible with the Lemma \ref{lem:com}, since if all such triangles commute, so does the big ones consisting of the amplitude maps only.

For a chain-complete $S$ all the Hilbert spaces on (\ref{startriangles}) are isomorphic. Indeed, $\mathcal{H}_{S|A} \cong \mathcal{H}_A$ since $A$ is not related to any other system, but also $\mathcal{H}_S \cong \mathcal{H}_{S|A}$ due to the chain-completeness. Hence, in such a case, all the isometries $\phi$'s need to be onto, so in fact, they are all unitary. Without this assumption, they are isometries for star-shaped fact-nets and partial isometries in general.

\passage{Maps as states}

The linear maps introduced above can be interpreted as states on the composite systems via the Choi-Jamiołkowski isomorphism. For example, the map $\phi_{SA}$ can be thought of as representing the correlation between the measuring apparatus $A$ and the quantum system $S$. Indeed, being unitary as maps, the states corresponding to them are maximally entangled
\begin{equation}
     \ket{\phi}_{SA} = \sum_{a \in F_{SA}}\ket{a}_A \otimes \ket{a}_S.
\end{equation}

The amplitude maps can also be understood similarly. For example, the state corresponding to the amplitude map $W^S_{AB}$ reads
\begin{equation}
    \begin{split}
         \ket{W}^S_{AB} &= \sum_{b \in F_{AB}}\ket{b}_B \otimes W^S_{AB}\ket{b}_B \\
         &= \sum_{a,b} W_S(b,a) \ket{b}_B \otimes \ket{a}_A
    \end{split}
\end{equation}
The amplitude, that expresses a correlation between the fact-sets $F_{SA}$ and $F_{SB}$, is here captured as an entangled state on $\mathcal{H}_{S|A} \otimes \mathcal{H}_{S|B}$ relative to the system $S$. Indeed, the amplitude map is exactly the set of equations relating the basis vectors of the relative Hilbert spaces, like in eq.\eqref{ketb}. The interpretation of the amplitude as a kind of relative state will be illustrated further below in the context of quantum reference frames.

\section{Quantum reference frames}
\label{sec:QRF}

So far, our examples were focused on star-shaped fact-nets. Such fact-nets enable a clear distinction between the centre and the surrounding systems. The centre is either seen as a quantum system surrounded by classical measurement devices or alternatively, the centre is a classical observer, surrounded by quantum systems. Such cases encompass most situations that were so far experienced in laboratories.  However, if one believes that any quantum system can serve as a reference system, one is led to also consider different situations, like the following triangle-shaped fact-net:
\begin{equation}\label{diag:qrf}
\begin{tikzcd}
 & A \arrow[rd, no head] \arrow[rd, no head, shift right] \arrow[rd, no head, shift left] &   \\
B \arrow[ru, no head, shift right] \arrow[ru, no head] \arrow[ru, no head, shift left] \arrow[rr, no head] \arrow[rr, no head, shift right] \arrow[rr, no head, shift left] &              & C
\end{tikzcd}
\end{equation}
Such fact-nets open up the possibility to compare different perspectives, which is the subject of quantum reference frames (QRF). A concrete question to ask is: given a state for the joint system $BC$ relative to $A$, what can be said about the state of $AB$ relative to $C$?

Although the subject is already old, the interest in it was renewed in \cite{giacomini2019} where a concrete proposal was made to describe a change of quantum reference frame. The work was later extended in \cite{delaHamette2020quantumreference} with a group-theoretic formalism. To address the issue in the fact-net formalism, let's first provide an example taken from \cite{delaHamette2020quantumreference}.

\passage{Example: three qubits}

Let's consider three qubits, like three 1/2-spins. In standard quantum mechanics, one would describe their states as up $\ket{\uparrow}$ or down $\ket{\downarrow}$. But in relational quantum mechanics, such an orientation up/down is relative to another orientation, so that the relative facts are `aligned' or `anti-aligned' that we denote respectively 0 and $\pi$. The fact 0 could correspond to $\uparrow \uparrow$ or $\downarrow \downarrow$. So we have the following triangle fact-net:
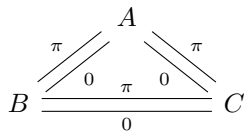
\begin{center}
\begin{tikzcd}
 & A \arrow[rd, "0"', no head, shift right] \arrow[rd, "\pi", no head, shift left] &   \\
B \arrow[ru, "0"', no head, shift right] \arrow[ru, "\pi", no head, shift left]  \arrow[rr, "0"', no head, shift right] \arrow[rr, "\pi", no head, shift left] &  & C
\end{tikzcd}
\end{center}

On the one hand, the amplitude $W_A$ gives a state of $BC$ relative to $A$. Take for instance
\begin{equation}
\label{eq:WA}
    \ket{W}^A = (\ket{0}^A_B + \ket{\pi}^A_B) \otimes \ket{0}^A_C. 
\end{equation}
It can also be written as an amplitude map $W^A_{BC} : \mathcal{H}_{C|A} \to \mathcal{H}_{B|A}$:
\begin{equation}
    W^A_{BC} = 
    \begin{pmatrix}
    1 & 0 \\
    1 & 0
    \end{pmatrix}
\end{equation}

On the other hand, the amplitude $W_B$ gives the state of $AC$ relative to $B$. By a sort of transitivity argument and a linearity principle, the QRF literature argues that $W_B$ is not free, but constrained and deducible from $W_A$. The argument is as follows:
\begin{enumerate}
    \item A state $\ket{0}^A_B \ket{0}^A_C$ means that $B$ is aligned with $A$ which is itself aligned with $C$, so that $B$ is aligned with $C$ and so the corresponding state relative to $B$ is $\ket{0}^B_A \ket{0}^B_C$.
    \item Similarly, a state $\ket{\pi}^A_B \ket{0}^A_C$, implies that $B$ is anti-aligned with $C$, so that the corresponding state relative to $B$ is $\ket{\pi}^B_A \ket{\pi}^B_C$.
    \item By requiring linearity, the state $W_B$, deduced from \eqref{eq:WA}, must be
\begin{equation}
\label{eq:WB}
 \ket{W}^B =  \ket{0}^B_A \ket{0}^B_C + \ket{\pi}^B_A \ket{\pi}^B_C. 
\end{equation}
\end{enumerate}
As a map the latter reads:
\begin{equation}
    W^B_{AC} = 
    \begin{pmatrix}
    1 & 0 \\
    0 & 1
    \end{pmatrix}.
\end{equation}
The surprising conclusion of this first computation is that a system $B$, in a superposed state relative to $A$ (eq. \eqref{eq:WA}), perceives $A$ in a state entangled with the rest of the world (eq. \eqref{eq:WB}). Similarly, one can deduce the view from $C$:
\begin{equation}
     \ket{W}^C =  \ket{0}^C_A \otimes (\ket{0}^C_B + \ket{\pi}^C_B). 
\end{equation}
or as a map
\begin{equation}
    W^C_{AB} = 
    \begin{pmatrix}
    1 & 1 \\
    0 & 0
    \end{pmatrix}.
\end{equation}
The consistency of the three views is expressed by the relation
\begin{equation}
    W^A_{BC} = W^C_{BA} W^B_{AC}
\end{equation}
which means the commutation of the triangle
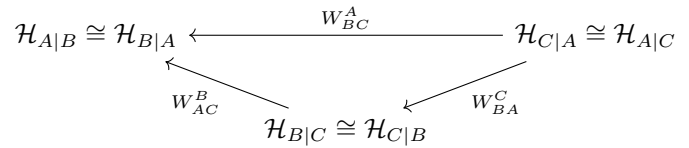
\begin{center}
\begin{tikzcd}
  \mathcal{H}_{A|B} \cong \mathcal{H}_{B|A}  &     & \arrow[ll, "W^A_{BC}"']  \mathcal{H}_{C|A} \cong \mathcal{H}_{A|C} \arrow[ld, "W^C_{BA}"] \\
&  \mathcal{H}_{B|C} \cong \mathcal{H}_{C|B} \arrow[ul, "W^B_{AC}"]    & 
\end{tikzcd}
\end{center}
With the Choi-Jamilkowski isomorphism $\mathbf{CJ}$, the relation reads
\begin{equation}
    \ket{W}^B = \mathbf{CJ}[\mathbf{JC}(\ket{W}^A) \circ \mathbf{JC}( \ket{W}^C )].
\end{equation}

\passage{Change of QRF}

Let us generalize the example by considering a triangle fact-net (ABC) with general fact-sets. Suppose that $W^A$ gives the state
\begin{equation}
    \ket{f}^A_B \ket{h}^A_C
\end{equation}
with $f \in F_{AB}$ and $h \in F_{AC}$. The fact $f$ is common to $A$ and $B$, but it can be viewed from two perspectives as $\ket{f}^A_B$ or $\ket{f}^B_A$. So, when we switch to the perspective of $B$, the relative state becomes
\begin{equation}
    \ket{f}^B_A \ket{g}^B_C
\end{equation}
with some $g \in F_{BC}$. 

We assume that the fact $g$ is deducible from $f$ and $h$. Therefore, there must be a function $\mathcal{G}_A$ such that $g = \mathcal{G}_A(h,f)$. Similarly, there must exist a function $\mathcal{G}_B(f,g)$ and $\mathcal{G}_C(g,h)$. For the consistency of the picture, when we switch to the view of $C$, we must recover the fact $h$, so we must have conditions such as
\begin{equation}
    \label{eq:consistency}
    f = \mathcal{G}_C(g,\mathcal{G}_B(f,g)).
\end{equation}

\passage{Fact-groups}

In \cite{delaHamette2020quantumreference}, a concrete proposal was made to specify the operators $\mathcal{G}_A, \mathcal{G}_B, \mathcal{G}_C$ above. It is done in the special case where the fact-sets all carry the same group structure $G$. More precisely, the group structure is given by bijective maps which send facts to group elements:
\begin{equation}
    \mathcal{I}^A_B : F_{AB} \to G.
\end{equation}
The maps satisfy the property that $\mathcal{I}^B_A(f) = \mathcal{I}^A_B(f)^{-1}$. So, one fact $f \in F_{AB}$ gives rise to two elements of the group, $\mathcal{I}^B_A(f)$ and $\mathcal{I}^A_B(f)$, inverse to each other, which can be understood as transformations to match the perspectives of $A$ to $B$, and vice versa.

Then, all relative Hilbert spaces are isomorphic to $L^2(G)$. In the three qubits example above, the group was $Z_2$ and the Hilbert space $\mathbb{C}^2$. By convention, one associates the neutral element $e$ to the fact between a system and itself. So for instance, a factorized state relative to $A$ reads
\begin{equation}
    \ket{e}^A_A \ket{f}^A_B \ket{h}^A_C.
\end{equation}
Then, the view from $B$ must be of the form
\begin{equation}
    \ket{f}^B_A \ket{e}^B_B \ket{g}^B_C,
\end{equation}
where $g$ must be a function of $f$  and $h$. In \cite{delaHamette2020quantumreference}, the following transformation is advocated:
\begin{equation}
   g = \left(\mathcal{I}^B_C\right)^{-1}( \mathcal{I}^A_C(h) \cdot \mathcal{I}^B_A(f) ).
\end{equation}
This defines the change of QRF from $A$ to $B$ on a basis. By linearity, it suffices to define a general operator
\begin{equation}
    U^{A\to B} : \mathcal{H}_{B|A} \otimes \mathcal{H}_{C|A}  \to  \mathcal{H}_{A|B} \otimes \mathcal{H}_{C|B} 
\end{equation}
such that
\begin{equation}
     U^{A\to B} \left( \ket{f}^A_B \ket{h}^A_C \right) = \ket{f}^B_A \ket{\left(\mathcal{I}^B_C\right)^{-1}( \mathcal{I}^A_C(h) \mathcal{I}^B_A(f) )}^B_C.
\end{equation}

The notation used here is different from that used in \cite{delaHamette2020quantumreference}. Indeed, the variable $g$ that labels the ket $\ket{g}$ has a slightly different meaning. For us, $\ket{g}^A_B$ must become $\ket{g}^B_A$ because $g$ is a fact shared symmetrically by $A$ and $B$. In \cite{delaHamette2020quantumreference} instead, $\ket{g}^A_B$ must correspond to $\ket{g^{-1}}^B_A$ because $g$ labels a transformation from $A$ to $B$. The relation between the two notations is given by the maps $\mathcal{I}^A_B$ which turn facts into transformations.

\section{Conclusion}

In this paper, we introduced the fact-nets formalism as a proposal for a mathematical framework fitted to the relational interpretation of quantum mechanics. The main points to take from the fact-nets are:
\begin{itemize}
    \item Hilbert spaces are not postulated primarily but derived from relative facts. They come equipped with a preferred basis. All properties of systems are relational and derived from relative facts involving them.
    \item The relations between systems are symmetric: no fundamental distinction between an observer and the observed is being made.
    \item The same fact-net with the same amplitude can describe very different experimental situations, depending on which system is taken as the observer. So fact-nets enable to express and reveal formal analogies.
    \item The correlation between facts is encoded in complex amplitudes. They are used to compute probabilities, which are always conditional, by design.
    \item Measurement appears as a restriction of fact-sets accompanied by a local ordering of systems. It is understood as a restriction of a theoretical fact-net to adjust for the physical setup that is being realized in time, in the lab.
    \item A fact-net of finite systems can be represented as a diagram with finite-dimensional Hilbert spaces and maps understood as entangled states.
    \item Relative states, quantum reference frame and changes of perspectives are natural concepts within the formalism.
\end{itemize}

We only sketched here the premiss of a theory of fact-nets. It motivates some lines of research to be developed in the short or medium term:
\begin{itemize}
    \item We envisage that time can be accommodated in the fact-net formalism as a labelling for nodes that connect to a system, each connection corresponding to different instance the system interacts with its surroundings. (This was illustrated in the propagator example,  Eq.~\ref{eq:propagator_3}.)

    \item The Wigner's friend scenario, non-locality, and other quandaries should be formulated and analysed in the fact-net framework, hoping some light could be shed on them. 

    \item Our focus on facts between systems instead of the systems themselves has a clear category-theoretic inspiration. Further investigations toward a representation theory of fact-nets as diagrams in the category of Hilbert spaces and partial isometries are currently carried out in the hope to close the gap with previous attempts of Yang \cite{yang2018} to replace quantum states by relation matrices between states.
    
    \item Our construction of the self-space of a system has only been thoroughly presented in the finite case, and still has to be extended to infinite fact-sets equipped with measurable structures.
    
    \item The information-theoretic flavour of the relational interpretation is not yet reflected in the fact-net formalism. However, one can already see how the structure of a fact-net allows systems to acquire varying amounts of information on one another. A formal correspondence involving von Neumann entropy still has to be figured out.
    
    \item Through examples, we have already seen that quantum reference frames are a promising field of application of fact-nets. Extensions of fact-nets to situations involving gauge symmetries are under development, and the links with the perspective neutral approach \cite{de2021perspective} and the measurement-theoretic setup of \cite{Loveridge2016RelativityOQ} are investigated. Dropping the assumption of the incompatibility of parallel facts seems a promising move towards encompassing general POVM and imperfect quantum reference frames.
    
    \item Finally, the emergence of classicality and causality in fact-nets still has to be clarified, notably by a more in-depth study of the chain property.
    
\end{itemize}

By developing fact-nets, our main aim is to clarify the debate on the relational interpretation, allowing it to take place on a more formalized ground. We hope that this framework will be as relevant for those trying to sharpen devastating no-go theorems for RQM as for those who desire to establish the relational interpretation as a matter of facts.

\begin{acknowledgements}

This research stems from discussions that occurred at the Sejny Summer Institute in July 2021, organized by the Basic Research Community for Physics. We must mention the essential contribution of Leon Loveridge and Anne-Catherine de la Hamette, and the helpful discussions with Andrea Di Biagio.

We also thank Alexandra Elbakyan for her help in accessing the scientific literature.

PMD is supported by the ID\# 61466 grant from the John Templeton Foundation, as part of the project \textit{Quantum Information Structure of Spacetime} (QISS). (\hyperlink{https://www.qiss.fr}{qiss.fr}). The opinions expressed in this publication are those of the authors and do not necessarily reflect the views of the John Templeton Foundation.

JG is supported by the Polish National Science Centre (NCN) through the OPUS Grant No. 2017/27/B/ST2/02959.

This work was partly supported by the Czech Science Foundation,
Grant No. GA\v{C}R 19-15744Y.

\end{acknowledgements}

\bibliographystyle{scipost}
\bibliography{facts.bib}

\end{document}